\documentclass[sigconf]{acmart}

\renewcommand\footnotetextcopyrightpermission[1]{}

\usepackage{subcaption}

\AtBeginDocument{%
  \providecommand\BibTeX{{%
    \normalfont B\kern-0.5em{\scshape i\kern-0.25em b}\kern-0.8em\TeX}}}

 \setcopyright{none}

\begin{document}

\title{Data-driven Identification of Number of Unreported Cases for COVID-19: Bounds and Limitations}

\author{Ajitesh Srivastava}
\email{ajiteshs@usc.edu}
\author{Viktor K. Prasanna}
\email{prasanna@usc.edu}
\affiliation{%
  \institution{University of Southern California}
  \city{Los Angeles}
  \state{CA, USA}
}

\renewcommand{\shortauthors}{Srivastava and Prasanna}

\begin{abstract}
  Accurate forecasts for COVID-19 are necessary for better preparedness and resource management. Specifically, deciding the response over months or several months requires accurate long-term forecasts which is particularly challenging as the model errors accumulate with time. A critical factor that can hinder accurate long-term forecasts, is the number of unreported/asymptomatic cases. While there have been early serology tests to estimate this number, more tests need to be conducted for more reliable results. To identify the number of unreported/asymptomatic cases, we take an epidemiology data-driven approach. We show that we can identify lower bounds on this ratio or upper bound on actual cases as a factor of reported cases. To do so, we propose an extension of our prior heterogeneous infection rate model, incorporating unreported/asymptomatic cases. We prove that the number of unreported cases can be reliably estimated only from a certain time period of the epidemic data. In doing so, we construct an algorithm called Fixed Infection Rate method, which identifies a reliable bound on the learned ratio. We also propose two heuristics to learn this ratio and show their effectiveness on simulated data. We use our approaches to identify the upper bounds on the ratio of actual to reported cases for New York City and several US states.  Our results demonstrate with high confidence that the actual number of cases cannot be more than 35 times in New York, 40 times in Illinois, 38 times in Massachusetts and 29 times in New Jersey, than the reported cases.
\end{abstract}




\begin{teaserfigure}
\centering
  \includegraphics[width=0.7\textwidth]{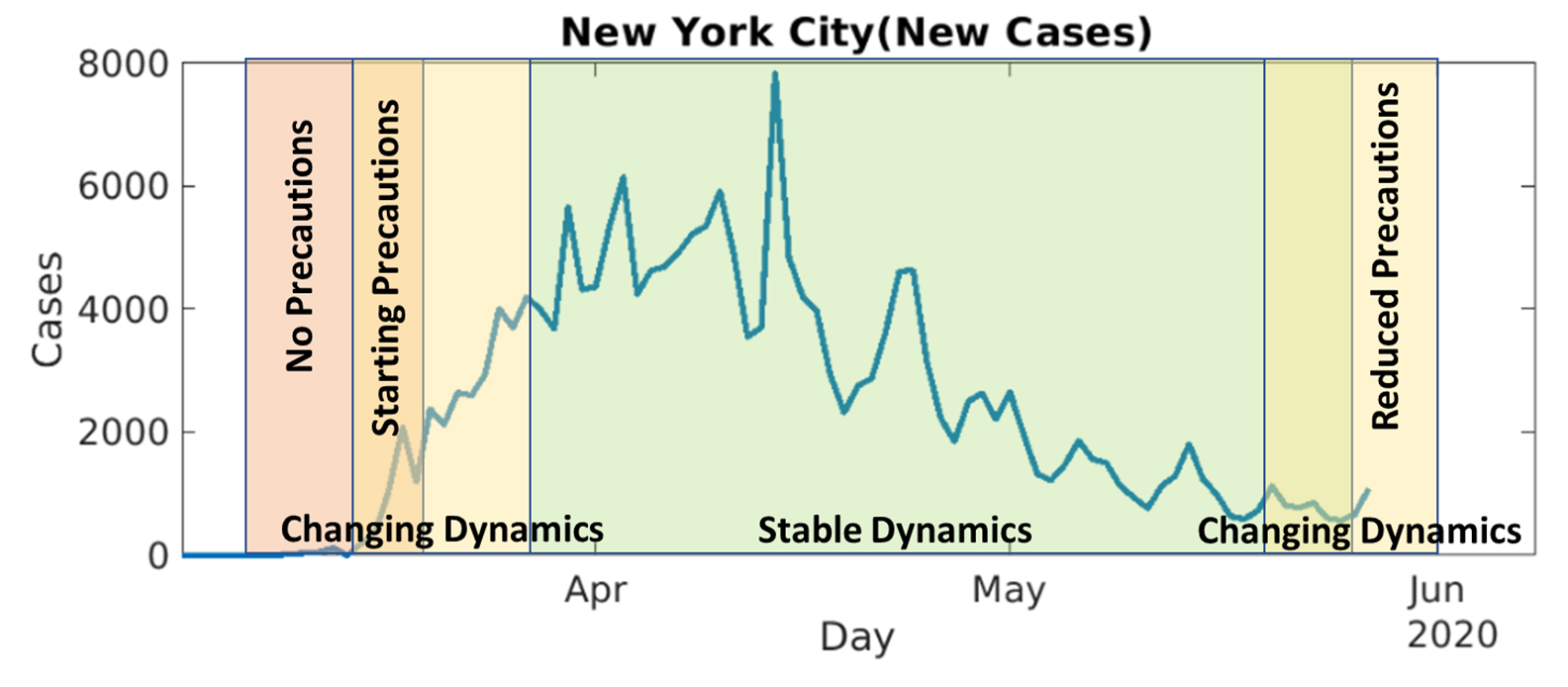}
  \caption{The social distancing phase allows us to model COVID-19 in a way when the effect of latent unreported/asymptomatic cases can be reliably observed.}
  \label{fig:teaser}
\end{teaserfigure}

\maketitle
\pagestyle{plain}

\newcommand{\gbar}{\bar{\gamma}}

\section{Introduction}

During the current COVID-19 pandemic, researchers have attempted to estimate the number of cases that are not being reported using antibody tests~\cite{bendavid2020covid}. This number is useful as it dictates the number of susceptible individuals, which in turn affects the long-term dynamics of the epidemic. 

We take a data-driven approach to model the existence of unreported cases in terms of probability of a case being reported. Due to a long period of social distancing, the infection dynamics are `stable',i.e, the parameters that drive the number of cases can be assumed to be constant over the period. This is unlike the earlier phase when the world had just started taking precautions during which a single model with fixed parameters would not have been able to explain the trends. Using the data from this ``stable'' phase (see Figure~\ref{fig:teaser}) of social distancing phase and before the precautions are reduced, we may be able to observe the effect of unreported cases.
We demonstrate that the probability of reporting can be reliably obtained only from certain parts of the time-series.
This in turn provides an estimated upper bound on the number of total actual cases as a factor of number of reported cases. Particularly, we prove that the probability of reporting has a negligible effect on the trend of reported cases in the initial part of the epidemic. Therefore, during that period, we cannot reliably learn the reporting probability. On the other hand, we also prove that learned probability is not reliable using only the later phase of the epidemic. Thus, there is a certain time interval over which the learned bound on reporting probability is reliable. We leverage the fact that reporting probability has negligible effect on the initial part of the timeseries and significant impact in the later part to construct an algorithm termed Fixed Infection Rate method. Our method can guarantee that the obtained upper bound is close to the true upper-bound. We also propose two heuristics that attempt to learn this upper-bound without any guarantee.
While we can also attempt to identify this bound without relying on a `stable' phase using adaptive models~\cite{srivastava2020learning}, it will introduce more hyperparameters making our estimation less reliable.

We are learning a lower bound on reporting probability (and correspondingly, upper bound on the actual cases) because we can only measure the combined effect of probability of reporting and complete isolation (see Section~\ref{sec:gbar}). This complete isolation is different from reducing social interactions. Reduced social interactions reduces the probability of a randomly selected infected person affecting a randomly selected susceptible person. On the other hand, complete isolation implies that a part of the population is removed and does not participate in the epidemic, effectively reducing the population by a constant factor. Since this factor is not known, we can only obtain a lower bound on reporting probability or an upper bound on the total cases as a factor of reported cases.

We proceed with an extension of the model proposed in~\cite{srivastava2020learning} which has been shown to perform accurate forecasts. We have previously used a preliminary version of this model in the DARPA Chikungunya forecasting challenge~\cite{darpachanllenge}, where we were one of the winners~\cite{darpawinners}. However, our approach for identifying the right data to reliably learn reporting probability may be applicable to other epidemiological models as well.
Our contributions are the following:
\begin{itemize}
    \item We propose an extension of our prior heterogeneous infection rate model that incorporates unreported/asymptomatic cases in the form of a parameter that measures the ratio of reported cases to actual number of cases.
    \item We prove that a bound on number of unreported cases can be reliably estimated only from certain data. 
    \item We propose Fixed Infection Rate Learning, an algorithm that leverages the effect of data on the model parameters to reliably identify a lower bound on reporting probability (and correspondingly, upper bound on actual cases as a factor of reported cases).
    \item We also propose two heuristics -- Non-linear Incremental Learning and Non-linear Curve Fitting, that attempt to learn a lower bound on reporting probability, but do not provide reliability guarantees.
    \item On simulated data, we show that our proposed method and proposed heuristics are accurately able to retrieve the ratio of reported to actual cases.
    \item We use our approaches to identify the lower bounds on the ratio of reported to actual cases for New York City and several US states.Our results demonstrate with high confidence that the actual number of cases are cannot be more than 35 times in New York, 40 times in Illinois, 38 times in Massachusetts and 29 times in New Jersey, than the reported cases.
\end{itemize}

\section{Related Work}

\subsection{Modeling Unreported Cases}
Several works in the literature~\cite{magal2018parameter,ducrot2020identifying,liu2020understanding} have attempted to model unreported cases by adding states such as asymptomatic and unreported to the Susceptible-Infected-Removed (SIR) model~\cite{bjornstad2002dynamics}.
Magal and Webb~\cite{magal2018parameter} propose a methodology for SIR model, that can determine the probability of reporting. This approach assumes that the `turning point', i.e., the time at which the number of new cases peaks, is known.
Ducrot et. al.~\cite{ducrot2020identifying} propose a method for identification of unreported cases from reported cases when the model parameters satisfy certain properties in an extension of SIR model. 
Liu et. al.~\cite{liu2020understanding} use a similar model but do not discuss the learnability of parameters related to asymptomatic and unreported cases.

\subsection{The SI-kJ$\alpha$ Model}
In~\cite{srivastava2020learning}, we proposed the SI-kJ$\alpha$ model 
for the spread of a virus like COVID-19 across the world which captures (i) temporally varying infection rates (ii) arbitrary regions, and (iii) human mobility patterns. Within every region (hospital/city/state/country), an individual can exist in either one of two states: susceptible and infected. A susceptible individual gets infected when in contact with an infected individual at a rate depending on when that individual got infected, i.e., rate of infection  is $\beta_1$ for an individual infected between $t-1$ and $t-J$, $\beta_2$ for an individual infected between $t-J$ and $t-2J$, and so on, thus resulting in $k$ sub-states of infection.
$J$ is a hyperparameters introduced for a smoothing effect to deal with noisy data.
It also avoids overfitting the model by using a small $k$ to capture dependency on the last $kJ$ days.
The hypothesis is that how actively one passes on the infection is affected by when they get infected. 
We assume that after being infected for a certain time, individuals no longer spread the infection, i.e., $\exists k$, such that $\beta_i = 0 \forall i>k$.

Also, people traveling from other regions can increase the number of infections in a given region. We assume that this infection can happen because of human mobility. Suppose $F(q, p)$ represents mobility from region $q$ to region $p$. 
Our model is represented by the following system of equations.

\begin{align}
\Delta S_t^p &= - \frac{S_{t}^p}{N^p} \sum_{i=1}^k \beta_i^p \Delta I_{t-i}^p \,,\label{eqn:delS}\\
\Delta I_t^p &= \frac{S_{t}^p}{N^p} \sum_{i=1}^k \beta_i^p (I_{t-(i-1)J}^p -I_{t-iJ}^p)  \nonumber \\
&+\delta \sum_q F(q, p) \frac{\sum_{i=1}^k \beta_i^q (I_{t-(i-1)J}^q -I_{t-iJ}^q)}{N^q}\,. \label{eqn:delI_kJ} 
\end{align}
Here, $S_t^p$ and $I_t^p$ represent the number of susceptible individuals and infected individuals respectively in the region $p$ at time $t$. 
Parameter $\delta$ captures the influence of passengers coming into the region.

Note that if we set $k=1, J = \infty$, and ignore mobility ($\delta =  0$), this reduces to Susceptible-Infected (SI) model~\cite{zhou2006behaviors}. On the other hand, with bounded $k=1$ and $J < \infty$, the model is a variation of Suceptible-Infected-Released/Recovered (SIR) model~\cite{bjornstad2002dynamics}, where an infected individual is active for $J$ units of time.

\section{Modeling Unreported Cases}

While unreported cases are not observed in the data, they affect the long term dynamics by infecting other individuals and by also reducing the number of susceptible individuals. 

The individuals who are never accounted for in the reporting (in the past or the future) can be classified into two categories: (i) unreported cases - those who get infected over the course of the epidemic but do not report it; and (ii) immune/isolated cases - those who have the antibodies without being infected during the epidemic or those who are completely isolated and have 0 probability of getting infected. For unreported cases, we can add another state to our model: An individual in the $i^{th}$ ``infected'' sub-state will be reported with probability $\gamma_i^p$. Thus, the total number of new reported cases is given by $\Delta R_t^p  = \sum_{i=1}^k \gamma_i^p (I_{t-(i-1)J}^p - I_{t-iJ}^p )$. Then the parameters will be learned by fitting the reported cases to $R_t^p$. The immune/isolated cases can be modeled as considering them not-susceptible, and hence not involved in the epidemic. This effectively reduces the size of the population considered for epidemic modeling.
Suppose, $\rho^p$ is the probability of a randomly selected individual in region $p$ to be immune/isolated. Then the number of susceptible individuals at time $t$ is given by $S_t^p = (1-\rho^p)N^p - I_t^p$, and $(1-\rho^p)N^p$ represents the reduced size of the population.

\subsection{Model Simplifications for Social Distancing}\label{sec:gbar}
In the period of social distancing, we assume that majority of the spread is ``community spread'' and infections due to travel across the regions (state/counties) can be ignored. For ease of notation, we drop the superscript $p$. For simplicity, we assume that $\gamma_i = \gamma, \forall i$. Further, we redefine $I_t$ to be the cumulative cases that could have been reported at time $t$ and $R_t$ to be the cases actually reported. This allows us to ignore explicit modeling of reporting delays. Therefore, we have
\begin{align}
\label{eqn:reported}
\Delta R_t  &= \gamma \sum_{i=1}^k (I_{t-(i-1)J} - I_{t-iJ})\nonumber \\
\mbox{And } R_t  &= \gamma I_t\,.
\end{align}

Combining Equation~\ref{eqn:reported} with Equation~\ref{eqn:delI_kJ} without the travel spread and adjusted population size, we get:
\begin{align}
\label{eqn:main}
    \frac{\Delta R_t}{\gamma} &= \frac{S_{t}}{(1-\rho)N} \sum_{i=1}^k \beta_i^p \frac{(I_{t-iJ} -I_{t-(i-1)J})}{\gamma} \nonumber \\
    \implies \Delta R_t &= \frac{(1-\rho)N - R_{t}/\gamma}{(1-\rho)N} \sum_{i=1}^k \beta_i^p (R_{t-iJ} - R_{t-(i-1)J}) \nonumber \\
    \implies \Delta R_t &= \left( 1 - \frac{R_{t}}{\gamma(1-\rho)N} \right) \sum_{i=1}^k \beta_i^p (R_{t-iJ} - R_{t-(i-1)J})
\end{align}

Equation~\ref{eqn:main} implies that only using the reported cases, the impact of $\gamma$ and $\delta$ cannot be separately measured. Setting $\bar{\gamma}  = \gamma(1-\rho) \leq \gamma$, we can identify a lower bound on $\gamma$. Note that $\gamma$ and $\rho$ are not separately needed to be able to forecast the number of reported cases, and knowing $\gbar$ is enough. However, this applies only when the infection dynamics are not changing. In the future, as the social distancing policies are relaxed, $\rho$ is expected to change and approach $1$, while $\gamma$ may remain constant assuming enough testing availability. Therefore, we wish to learn $\gamma$ but at this point, we can only identify $\gbar$ which forms a lower bound for $\gamma$.

\begin{figure*}[!htpb]
  \centering
  \subcaptionbox{Initial Phase\label{fig:sim_pre}}{
    \includegraphics[width=0.31\textwidth]{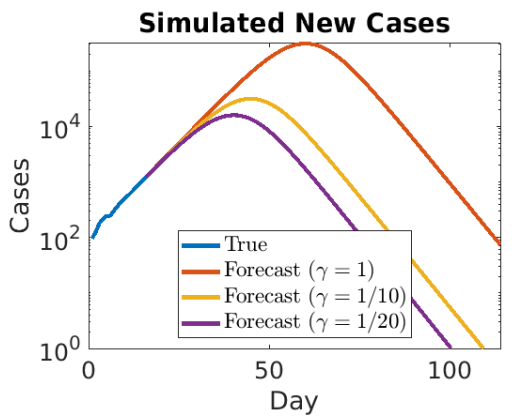}
  }
  \subcaptionbox{Around Peak\label{fig:sim_peak}}{%
    \includegraphics[width=0.31\textwidth]{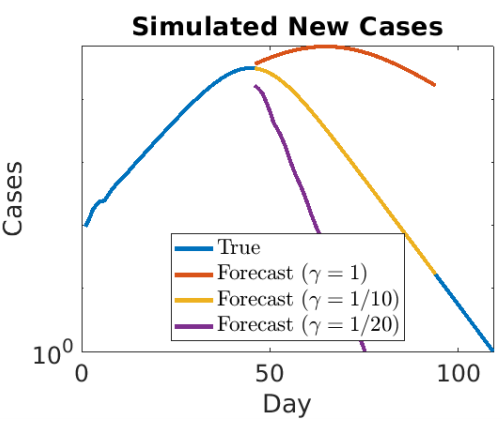}
  }
  \subcaptionbox{`Tail' Phase\label{fig:sim_post}}{%
    \includegraphics[width=0.33\textwidth]{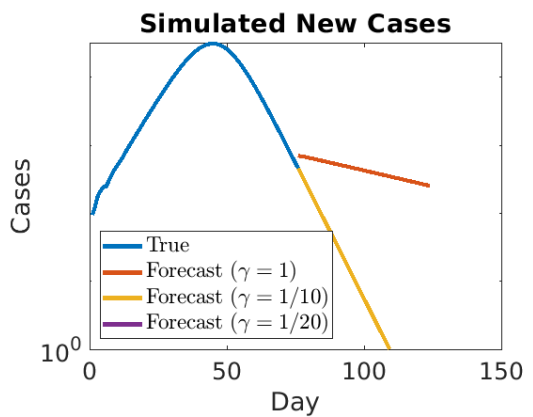}
  }
  \caption{Effect of varying $\gamma$ at different phases of the epidemic on the reported cases.}
  \label{fig:sim_sensitivity}
\end{figure*}

\subsection{Parameter Learnability}\label{sec:leanrability}

Let $\mathbf{\beta} = [\beta_1 \dots \beta_k]$, and $\mathbf{X_t} = [(R_{t} - R_{t-J}) \dots (R_{t-(k-1)J} - R_{t-kJ})]^T$. Sensitivity of $\Delta R$ with respect to $\gamma$ is

\begin{align}
\label{eqn:gamma_sens}
    \frac{\partial \Delta R_t}{\partial \gbar} = \frac{R_{t}}{\gbar^2 N} \mathbf{X_t}\mathbf{\beta}.
\end{align}
\begin{align}
\label{eqn:beta_sens}
    \frac{\partial \Delta R_t}{\partial \beta} = \left(1 - \frac{R_{t}}{\gbar N}\right) \mathbf{X_t}.
\end{align}
In the initial phase of the epidemic, $\frac{R_{t-1}}{N} \approx 0$. Therefore, Equation~\ref{eqn:gamma_sens} suggests that the number of reported cases is not sensitive to $\gbar$ in the initial phase of the epidemic, when $\frac{R_{t-1}}{N} \approx 0$. On the other hand, Equation~\ref{eqn:beta_sens} suggests that number of new reported cases is sensitive to $\beta$. 

Suppose, $\gbar^*$ is the true value and we train by ignoring the parameter, effectively setting it to $1$ to obtain $\mathbf{\beta_0}$. Then, we get the same timeseries, if $\forall t$,

\begin{align}
\label{eqn:fixed}
    \left( 1 - \frac{R_t}{\gbar^* N}\right) \mathbf{X_t}\mathbf{\beta^*} = \left( 1 - \frac{R_t}{N}\right) \mathbf{X_t}\mathbf{\beta_0} \nonumber \\
    \frac{\mathbf{X_t}\mathbf{\beta_0}}{\mathbf{X_t}\mathbf{\beta^*}} = 1 - \frac{R_t(1-\gbar)}{\gbar (N - R_t)}\,,
\end{align}
which is close to 1, when $R_t \ll N$. Figure~\ref{fig:sim_sensitivity} demonstrates this fact. We simulate an epidemic with $\beta = [0.4\,\, 0.2], N = 1,000,000$ and $\gbar = \gamma = 1/10$. We then attempt to ``forecast'' assuming the knowledge of $\beta$, and various values of $\gbar = \gamma =1,  1/10$ and $1/20$. Observe that in the initial phase of the epidemic (Figure~\ref{fig:sim_pre}) all three trends are similar until they get close to the peak. Starting at the peak (Figure~\ref{fig:sim_peak}) and after the peak (Figure~\ref{fig:sim_post}), with the same initial values and $\beta$, significantly different forecasts are obtained by varying $\gamma$. By setting $k=1$ in Equation~\ref{eqn:fixed}, the following can be easily proved.
\begin{theorem}
\label{thm:fixed}
For a given $R_f$,  $\exists \epsilon > 0$, such that $\forall  R_t \leq R_f$, $0 \leq \frac{\beta^* - \beta_0}{\beta^*} \leq \epsilon$.
\end{theorem}
\begin{proof}
From Equation~\ref{eqn:fixed}, easy to see that $\beta_0 \leq \beta^*$. Setting $\epsilon = \frac{R_f(1-\gbar^*)}{\gbar (N - R_f)}$ completes the proof.
\end{proof}

Theorem~\ref{thm:fixed} suggests that early part of the epidemic is not reliable for learning $\gbar$. However, this does not imply that we should always prefer a high value of $t$ in the following where we explore the effect of the ``tail" part of the epidemic on the learnability of $\gbar$.
\begin{lemma}
There exists $\tau$ such that $\gbar$ that describes the data for $R_t > R_{\tau}$ is not unique.
\end{lemma}
\begin{proof}
We prove this by showing that there is a $t_u$ such that  for $t > t_u$, there are at least two sets of parameter $(\beta_1, \gbar_1)$ and $(\beta_2, \gbar_2)$ that fit the data for $t > t_u$, i.e., the following has a feasible solution.
\begin{align*}
    \Delta R_t = \left(1 - \frac{R_t}{\gbar_1 N} \right)\mathbf{\beta_1}\mathbf{\Delta X_t} = \left(1 - \frac{R_t}{\gbar_2 N} \right)\mathbf{\beta_2}\mathbf{\Delta X_t}.
\end{align*}
Setting $k=1$, $\mathbf{\Delta X_t}$ becomes a scalar. After some algebraic manipulations, we get
\begin{align}
\label{eqn:gbar2}
    \gbar_2 = \frac{(\beta_2/\beta_1)\gbar_1 R_t}{R_t - (1 - \beta_2/\beta_1)\gbar_1 N}
\end{align}
This is a valid solution, if $0 < \gbar_2 \leq 1$. Without loss of generality, we can assume $\beta_2 < \beta_1$. Then
\begin{align*}
    \gbar_2 > 0 \implies R_t > \gbar_1 N(1 - (\beta_2/\beta_1)),\\
    \mbox{And }
    \gbar_2 \leq 1 \implies R_t > \gbar_1 \frac{N(1 - (\beta_2/\beta_1))}{1- (\beta_2/\beta_1)\gbar_1}.
\end{align*}
Therefore, if the data contains $R_t$ such that the above holds for all $t$, then at least two solutions for $(\beta, \gbar)$ exist.
\end{proof}
The above lemma suggests that we should not attempt to learn the parameters solely from the ``tail'' of the epidemic. However, using the beginning part only, we cannot reliably learn $\gbar$ as discussed earlier. Next, we identify what data needs to be included to guarantee accurate learning of $\gbar$.
\begin{theorem}\label{thm:upperT}
Suppose, $(\beta_0, \gbar_0)$ is a solution obtained from the given data. Let $\beta^*  \geq \beta_0 \geq (1-\epsilon)\beta^*$, for some $0 \leq \epsilon < 1$. Then for any $R_\tau$, there exists a $0<\delta<1$ such that choosing data $R_t > R_\tau$ guarantees that $(1-\delta)\gbar_0 \leq \gbar^* \leq \gbar_0$.
\end{theorem}
\begin{proof}
Since, $\beta_0 \leq \beta^*$, $\gbar_0 \geq \gbar^*$. Suppose, for some $\delta > 0$, we wish to prove that that $\gbar^* \geq  (1-\delta)\gbar_0$. Assume the contrary that $\gbar^* < (1-\delta)\gbar_0$.
Then, using Equation~\ref{eqn:gbar2},
\begin{align*}
    \gbar_0 &= \frac{\gbar^*(\beta_0/\beta^*)R_t/N}{R_t/N - (1-\beta_0/\beta^*)\gbar^*} \\
    \implies \frac{\gbar^*}{1-\delta} &< \frac{\gbar^*(\beta_0/\beta^*)R_t/N}{R_t/N - (1-\beta_0/\beta^*)\gbar^*}
\end{align*}

Using $\beta^* \geq \beta_0$ in the numerator and $\beta_0 \geq (1-\epsilon)\beta^*$ in the denominator of the RHS, we get
\begin{align*}
    \frac{1}{1-\delta} &< \frac{R_t/N}{R_t/N - \epsilon\gbar^*}\\
    \implies R_t &< \frac{N\epsilon \gbar^*}{\delta} \leq \frac{N\epsilon \gbar_0}{\delta}.
\end{align*}
Therefore, if $R_t \geq R_\tau = \frac{N\epsilon \gbar_0}{\delta}$ the above is not feasible, and thus $\gbar^* \geq  (1-\delta)\gbar_0$.
\end{proof}

Finally, we present how $\gbar$ affects the peak of the epidemic.
\begin{theorem}\label{thm:peak}
If the peak of new cases happens when the total cases are $R_{peak}$, then 
\begin{equation}
    \gbar \approx \frac{R_{peak}/N}{1 - \frac{1}{J\|\mathbf{\beta}\|_1}},
\end{equation}
where $\|\mathbf{\beta}\|_1 = \sum_i \beta_i$.
\end{theorem}
\begin{proof}
At the peak, we assume that $\Delta R_t$ remains constant for a window of $kJ + 1$ time steps, i.e., $\Delta R_t = r, \forall t = \tau, \tau - 1, \dots, \tau - kJ$. Then $\mathbf{\beta}\mathbf{X} = J\|\mathbf{\beta}\|_1$. Therefore, we have
\begin{align}
    r &\approx \left(1 - \frac{R_\tau}{\gbar N} \right)J\|\mathbf{\beta}\|_1 r \nonumber \\
    \implies  \gbar &\approx \frac{R_\tau/N}{1 - \frac{1}{J\|\mathbf{\beta}\|_1}}.
\end{align}
\end{proof}

Next, we utilize Theorems~\ref{thm:fixed},~\ref{thm:upperT}, and~\ref{thm:peak} to learn the parameters $\beta$ and $\gbar$.

\section{Learning}
Unlike~\cite{srivastava2020learning} where the goal was to perform forecasts in an adaptive fashion even during changing policies, here, our main goal is identifying $\gbar$. This knowledge can then be used for performing forecasts. For learning, we first manually identify and remove the part of the data where the effect of social distancing is visible. For instance, in Figure~\ref{fig:teaser} the initial part shows rapid rise when no precautions were taken. This step is necessary for our axiom that the remaining data can be assumed to follow the same dynamics, i.e, has a true unique $(\beta,\gbar)$.





\subsection{Fixed Infection Rate Method}

In this approach we utilize the fact that the effect of the unreported cases is not seen in the initial part of the infection. 
Therefore, we consider an initial part of the reported cases data up to time $t_f$. We use this initial part to train the model to learn $\beta_0$  by fixing $\gbar = 1$. Then, by Theorem~\ref{thm:fixed}, $\beta^* \leq \beta_0 \leq (1-\epsilon)\beta^*$, for some $\epsilon$.

Then, we train a linear model by fixing the previously learned $\beta_0$ as a constant and learn $\gbar_0$.
We identify the largest value of $R_t = R_{max}$ available in the dataset. 
From Theorem~\ref{thm:upperT}, it follows that setting 
\begin{equation}\label{eqn:delta}
\delta = \frac{N\epsilon \gbar_0}{R_{max}}    
\end{equation}
 ensures that there is at least one data point for the model to identify $\gbar_0$ such that $(1-\delta)\gbar_0 \leq \gbar^*$. 

Next we discuss, how to identify the value of $\epsilon$. Note that $\epsilon$ as calculated in Theorem~\ref{thm:fixed} relies on $\gbar^*$, which is not known. We use the fact that a $\delta$ must exist such that $\gbar^* \geq (1-\delta)\gbar_0$. Using this bound in Theorem~\ref{thm:fixed}, there exists a $\delta$ for which
\begin{equation}\label{eqn:eps_bound}
    \epsilon = \frac{R_f\left(1 - {(1 - \delta)\gbar_0}\right)}{(N - R_f){(1-\delta)\gbar_0}}.
 \end{equation}
 Putting the value $\epsilon$ in Equation~\ref{eqn:delta} results in a quadratic equation in $\delta$ with the smaller root
 \begin{equation}\label{eqn:final_delta}
     \delta = \frac{1-\frac{R_f}{N} - \gbar_0 - \sqrt{\left(1-\frac{R_f}{N} - \gbar_0\right)^2 - 4\left(1-\frac{R_f}{N}\right)\left(\frac{R_f}{R_{max}}-\gbar_0\right)}}{2\left(1- \frac{R_f}{N}\right)}
 \end{equation}
 
 \begin{itemize}
     \item Test1 (hard): Is $\delta$ a real number and in $(0, 1)$? If not, then the method fails, as we are unable to guarantee reliability. 
     \item Test2 (soft): For a given $\delta_3 > 0$, and the number of cumulative reported cases at the peak $R_{peak}$, $(1-\delta_3)\gbar_0 \leq \frac{R_{peak}/N}{1 - 1/(J\|\beta_0\|)} \leq (1+\delta_3)\gbar_0$? This is a ``soft'' test in the sense that it is based on an approximation and can be performed only if the ``peak" is available. Identifying the actual peak is difficult due to noisy data, and thus $\delta_3$ provides a soft margin for the peak.
 \end{itemize}
 
The parameters are learned using least square estimation:
\begin{align}
    LSE &= \sum_{t=\tau}^{T} \left( \left( 1 - \frac{\hat{R_t}}{\gbar N}\right) \mathbf{X_t}\beta - \Delta \hat{R_t}\right)^2 \\
\end{align}
Here $\hat{R_t} \forall t$ are true observed values. Least square optimization is performed using trust-region reflective algorithm~\cite{coleman1996interior}. Note that the above approach may be prone to noisy initial values. However, we smooth the data before learning the parameters to avoid noise.

Alternatively, the initial values $\Delta R_{\tau-J}, \dots, \Delta R_{\tau-1}$ can also be treated as learnable parameters. In this case, we fit the curve obtained by the recurrence relation $\Delta R_t = \left(1 - \frac{R_t}{\gbar N}\right)X_t\beta$ to the observed data $<\Delta R_{\tau-J}, \dots \Delta R_{\tau-1}, R_{\tau}, R_{\tau+1}\dots \Delta R_{T}>$. While this approach is better for dealing with noisy data, it may be prone to overfitting due to additional $J$ parameters. Least square optimization is performed using trust-region reflective algorithm~\cite{coleman1996interior}. 

\subsection{Heuristic Methods}
We also propose treating $\beta$ and $\gbar$ simultaneously as learnable parameters as a heuristic approach. Since, Theorems~\ref{thm:fixed} and~\ref{thm:upperT} do not apply, we cannot perform Test1 to ensure reliability. However, we can perform Test2.
As in the case of Fixed Infection Rate Method, we have two ways of learning the heuristic models.

\paragraph{Non-linear Incremental Learning} 
The parameters are learned using least square estimation:
\begin{align}
    LSE &= \sum_{t=\tau}^{T} \left( \left( 1 - \frac{\hat{R_t}}{\gbar N}\right) \mathbf{X_t}\beta - \Delta \hat{R_t}\right)^2 \\
\end{align}

\paragraph{Non-linear Curve Fitting} Learning is performed by fitting a curve over time as opposed to a linear model by treating the initial values $\Delta R_{\tau-J}, \dots, \Delta R_{\tau-1}$ as learnable parameters as well.

It is possible to derive reliability bounds on these heuristics as well, however, they are unlikely to be useful. We wish to identify a lower bound on reporting probability, therefore, if $\gbar_0 \leq \gbar^*$, then we have nothing to prove. Suppose, $\gbar_0 \geq \gbar^*$. Then we would like to show that $\gbar^* \leq (1-\delta)\gbar_0$, which follows the same derivation as Equation~\ref{eqn:delta}. Choosing an epsilon here is difficult -- Using the scheme as in Fixed Infection Rate algorithm leads to $\delta>1$. A valid choice is $\epsilon = 1 -\beta_0$ (obtained using $\beta^* \leq 1$), which would result in $\delta = N(1-\beta_0)\gbar_0./R_{max}$. This is often larger than 1 in practice, and thus not useful. 

Here, we have chosen $k=1$ as our Test1 is derived for scalar $\mathbf{\beta}$. However, the above algorithms can be used (without reliability tests) for any value of $k$ with Test2. In Section~\ref{sec:sim} we have explored the effectiveness of the above algorithms for $k>1$. 

\section{Experiments}

\begin{figure*}[!htpb]
  \centering
  \subcaptionbox{(10, 30)\label{fig:first-subfig}}{
    \includegraphics[width=0.33\textwidth]{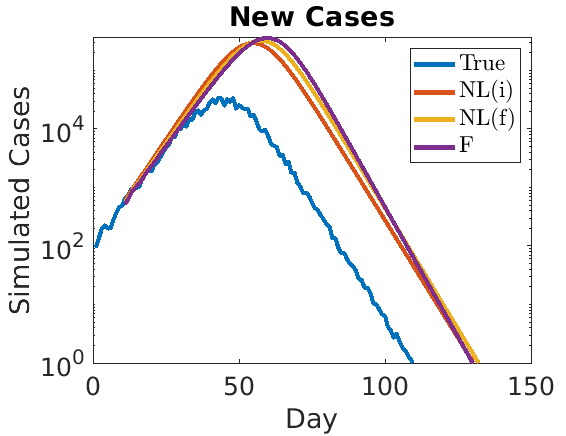}
  }
  \subcaptionbox{(30, 50)\label{fig:second-subfig}}{%
    \includegraphics[width=0.33\textwidth]{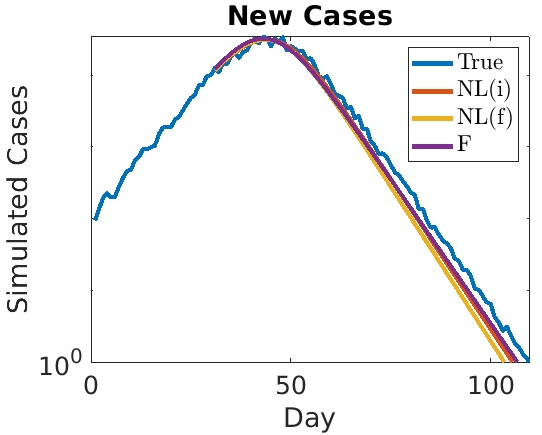}
  }
  
  \subcaptionbox{(50, 70)\label{fig:second-subfig}}{%
    \includegraphics[width=0.33\textwidth]{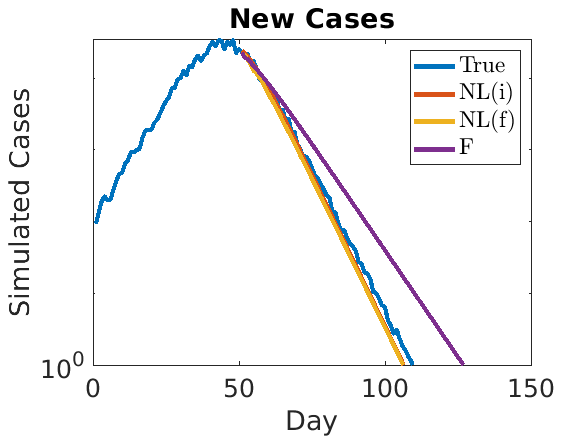}
  }
  \subcaptionbox{(70, 90)\label{second-subfig}}{%
    \includegraphics[width=0.33\textwidth]{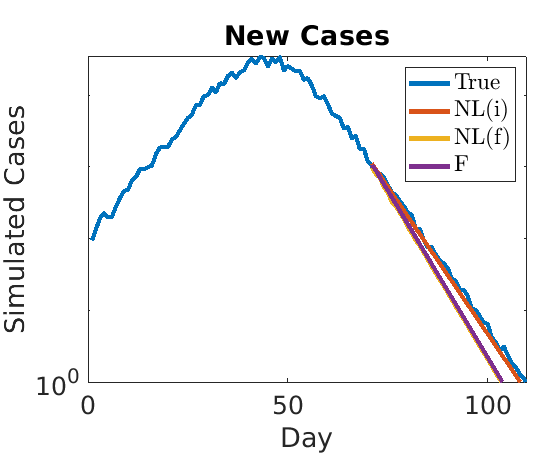}
  }
  \caption{Fitting the models over various intervals in the simulated data.}
  \label{fig:sim_sensitivity}
\end{figure*}

\subsection{Setup}
We obtained all the reported cases fom JHU CSSE COVID19 dataset~\cite{JHUdata}. Particularly we extracted county level data for New York City and Los Angeles. These were used because these two counties have performed serology tests with initial estimation of number of unreported cases. We further performed experiments on all US states, most of which did not pass our tests for reliability. Here we will report the results on New York, Illinois, Massachusetts, and New Jersey - four of the states with the most reported cases. Population of the counties and states were obtained from the US Census Bureau~\cite{USpopu}.

The county data showed significant amount of noise, and so, it was smoothed with moving average over two weeks, before applying our learning algorithms. The state-level timeseries were relatively less noisy, and were smoothed with moving average over one week. All the code was written in MATLAB and is available online\footnote{\url{https://github.com/scc-usc/ReCOVER-COVID-19}}. We set $k=1$ and $J=7$ for the US counties and states. The choice for $J$ was driven by observed weekly periodicity in the data~\cite{ricon2020seven}.

\begin{table*}[!htbp]
    \centering
       \caption{Learned parameters $(\beta_1, \beta_2)$, $\gbar$ from simulated experiments. The true value of $\gbar = 0.1$.}
    \begin{tabular}{|c|c|c|c|}
    \hline
    $(\tau_1, \tau_2)$ & NL(i) & NL(f) & F \\
    \hline
         (10, 30) &	(0.1723, 0.3619), 1 &	(0.3487, 0.02453), 1 &	(0.5569, 0.1071), 1 \\
(30, 50)&	(0.4408, 0.1793, \textbf{0.934} &	(0.4408, 0.1793), \textbf{0.092} &	(0.1750, 0.3620), \textbf{0.1095} \\
(50, 70)&	(0.2064, 0.4099), \textbf{0.1036} &	(0.5153, 0.0991), \textbf{0.0916} &	(0, 0.3652), 0.2440 \\
(70, 90)&	(0.0898, 0.0853), 1&	(0.1246, 0), 0.7796 &	(0, 0.3438), 0.1956 \\
\hline
    \end{tabular}
     \label{tab:sim}
\end{table*}

\subsection{Simulation}\label{sec:sim}
To demonstrate the effectiveness and limitations of the three approaches, we used the same setting as in Section~\ref{sec:leanrability} but with added noise to simulate an epidemic. We attempted to retrieve the parameters $(\beta, \gbar)$ using our three learning approaches - Non-linear Incremental Learning NL(i), Non-linear Curve Fitting NL(f), and Fixed Infection Rate Learning (F). These methods learn the models fitted on data for $T \in (\tau_1, \tau_2)$ for various intervals. Here, Fixed Infection Rate approach is simplified to use use $(1, \tau_1)$ to first identify $\beta$, and $(\tau_1, \tau_2)$ to identify $\gbar$, without any reliability guarantee. Note that here, $k=2$ and our reliability analysis applies only to $k=1$. Regardless, we perform these experiments to observe the effect of $k>1$.
Figure~\ref{fig:sim_sensitivity} shows the fit along with forecasts until the end of the epidemic. Observe that for the interval $(30,50)$ all methods are able to accurately forecast. 'NL(i)' and 'NL(f)' are able to accurately forecast for the interval $(50, 70)$. It also seems that the three methods accurately forecast by learning on the interval $(70, 90)$. To assess whether these fits actually retrieve the values of $\gbar$, we present the learned parameters in Table~\ref{tab:sim}. Note that only for the interval $(30, 50)$ all three methods are able to identify $\gbar$ close to its original value, i.e., $0.1$. While we were able to see accurate forecasts for the interval $(70, 90)$, the learned values of $\gbar$ are far from the true value. This reinforces our claim that there exists a certain window of data which is needed to accurately learn $\gbar$.

\begin{figure*}[htpb]
\captionsetup[subfigure]{labelformat=empty}
  \centering
  \subcaptionbox{}{
    \includegraphics[width=0.4\textwidth]{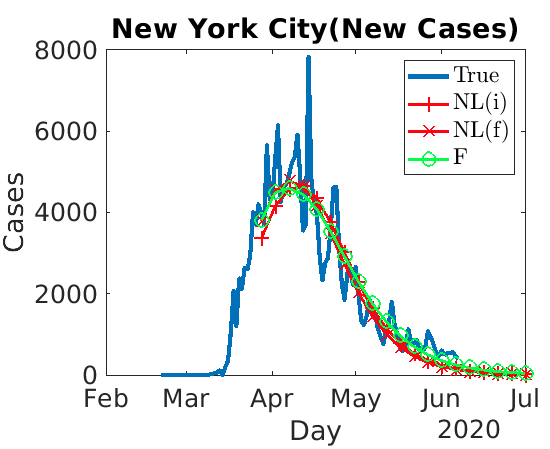}
  }
  \subcaptionbox{}{%
    \includegraphics[width=0.4\textwidth]{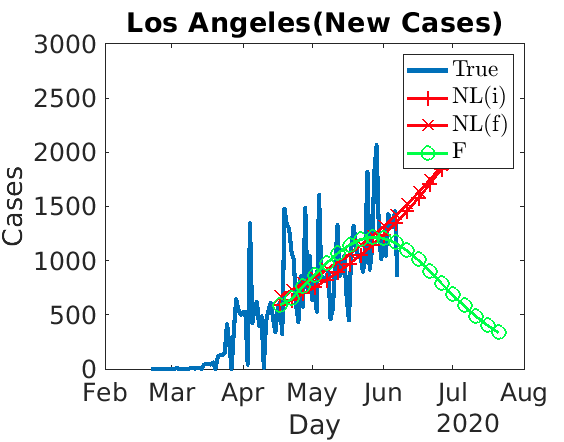}
  }
  \caption{Model fittings for counties using our three algorithms.}
  \label{fig:counties}
\end{figure*}

\begin{table*}[!htbp]
    \centering
    \caption{Estimated upper bound on number of total cases as a factor of reported cases for the counties.}
    \begin{tabular}{|l|c|c|c|c|}
    \hline
    States & NL(i) & NL(j) & F \\
    \hline
New York City & 41.2712 - 44.2499 - 47.6919 & 40.2654 - 42.8327 - 45.7496 & 38.8609 - 39.7612 - 62.9513\\
Los Angeles &  OOR &	OOR & (x)Test1 \\
\hline

    \end{tabular}
    \label{tab:counties}
\end{table*}

\subsection{Results: US Counties}
Figure~\ref{fig:counties} shows the model fit obtained on New York City and Los Angeles. Recall that $\gbar = (1-\rho)\gamma$, where $\gamma$ is the probability of reporting an infected case. Therefore, $1/\gbar$ forms the upper bound on the estimated number of total cases as a factor of reported cases.
We report these upper bounds in Table~\ref{tab:counties}. We have shown the factors obtained using $95\%$ confidence interval on $\gbar$. Additionally, for Fixed Infection Rate learning, we have provided an additional bound obtained from Theorem~\ref{thm:upperT} with $\delta$ obtained from Equation~\ref{eqn:final_delta}. 
The three methods result in factors close to each other (39-44) for New York City. However, the reliable bound obtained was $63$.
Figure~\ref{fig:counties} suggests that all three methods produce good fit for New York City. Note that this factor provides an upper bound on the actual ratio of total to reported cases. As an illustration, if we agree that the bound obtained for NYC is $40$ and $\rho = 0.5$, i.e., half of the population was able to completely isolate itself reducing its probability of infection to zero, then the the number of true cases will $0.5\times 36$, i.e, $18$ times of the reported cases.
On the other hand, none of the results for Los Angeles were sensible (see Table~\ref{tab:counties}). `OOR' indicates that the $95\%$ confidence interval was out of the feasible range of the solution. For the method `F', Test1 failed. It implies that it may be too early to reliably estimate the upper bound of this factor from Los Angeles data.

Note that antibody tests in New York in April estimated that $24.7\%$ of the entire population were infected\footnote{\url{https://www.livescience.com/covid-antibody-test-results-new-york-test.html}}. Based on the population of New York City and the number of reported cases at the time, this translates to actual cases being roughly $13.8$ times the reported cases.

\begin{figure*}[htpb]
  \centering
  \captionsetup[subfigure]{labelformat=empty}
  \subcaptionbox{}{
    \includegraphics[width=0.4\textwidth]{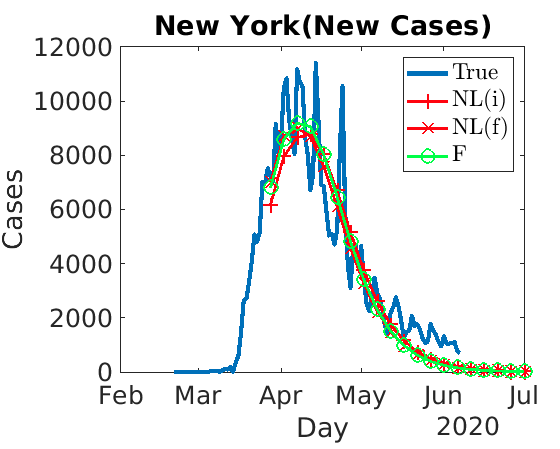}
  }
  \subcaptionbox{}{
    \includegraphics[width=0.4\textwidth]{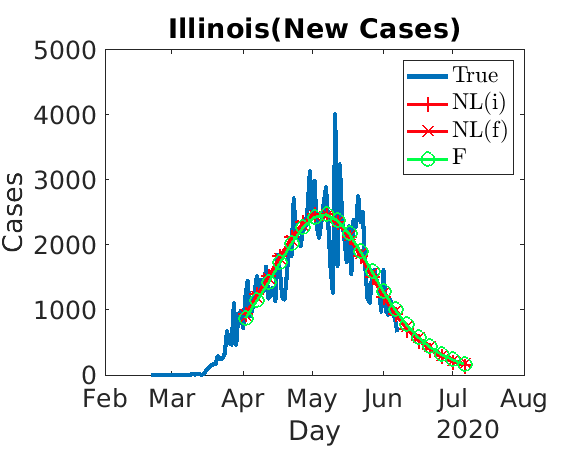}
  }
  \subcaptionbox{}{
    \includegraphics[width=0.4\textwidth]{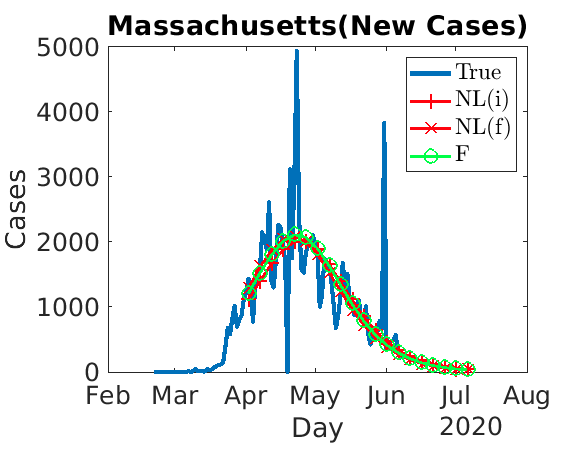}
  }
  \subcaptionbox{}{
    \includegraphics[width=0.4\textwidth]{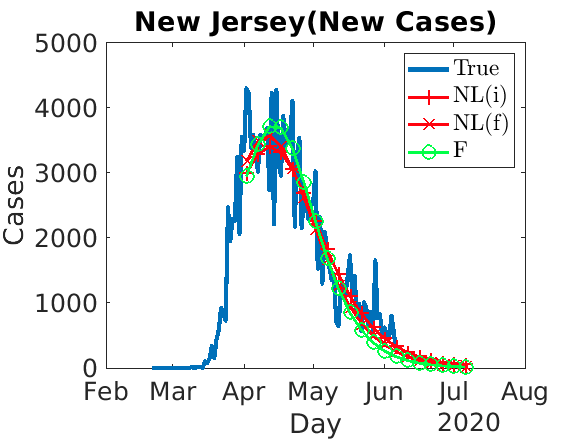}
  }
   \caption{Model fittings for states using our three algorithms.}
  \label{fig:states}
\end{figure*}

\begin{table*}[!htbp]
    \centering
    \caption{Estimated upper bound on number of total cases as a factor of reported cases.}
    \begin{tabular}{|l|c|c|c|c|}
    \hline
    States & NL(i) & NL(j) & F \\
    \hline
New York & 22.2565 - \textbf{24.8175} - 28.0445 & 21.1909 - \textbf{23.0848} - 25.3503 & 22.7918 - \textbf{23.4891} - 35.1732\\
Illinois & 30.3602 -\textbf{33.4322} - 37.1959 & 31.9536 - \textbf{33.4813} - 35.1624 & 36.2494 - \textbf{36.7681 }- 40.8638\\
Massachusetts& 24.1917 - \textbf{27.3885} - 31.5589 & 25.4437 - \textbf{27.5206} - 29.9668 & 30.3027 - \textbf{31.4906} - 38.2838\\
New Jersey & 18.5704 - \textbf{20.3788} - 22.5774 & 19.0749 - \textbf{20.3332} - 21.7692 & 18.6939 - \textbf{19.0698} - 29.2222\\
 \hline
    \end{tabular}
    \label{tab:states}
\end{table*}

\subsection{Results: US States}
We also estimated the bound on the total number of actual cases as a factor of reported cases for various states. Table~\ref{tab:states} shows the results for New York, Illinois, Massachusetts, and New Jersey. All the states presented here, passed Test1 and Test2. 
Figure~\ref{fig:states} shows the model fit obtained using the learned parameters.

For New York our methods estimated that the bound on total cases is 23-25 times of the reported cases with the reliable worst case bound being 35.17. Note that the state-wide antibodies study in early May  estimated that 12.3\% of the state population was infected\footnote{\url{https://www.governor.ny.gov/news/amid-ongoing-covid-19-pandemic-governor-cuomo-announces-results-completed-antibody-testing}}. This translates to actual cases being roughly 7.6 times the reported cases.
For Illinois, Massachusetts and New Jersey, this factor is roughly 33-37, 27-31, and 19-21, with worst case upper bound being 40.86, 38.28, and 29.22, respectively. If we assume that these states are similar enough that they have the same probability $\gamma$ of reporting and the same fraction of population that is completely isolated, then we can conclude that for all these states, the true cases cannot be more than 29.22 times, which satisfies all the upper bounds. All four states passed Test2 (peak test) with $\delta_3 = 0.2$.

We have not presented results for the US at country-level due to high heterogeneity in the infection trends of various states. Therefore, learning a single parameter for the entire country may not be accurate, and it may be better to learn separately for different states.

\section{Conclusions}
We have proposed Fixed Infection Rate algorithm to reliably estimate a bound on number of unreported cases. The algorithm is built upon key theorems that identify limitations of learnability of reporting probability. We have also proposed two heuristics that learn this bound but do not provide guarantees. We demonstrate through simulated experiments that all three methods are able to identify the bound correctly on certain regions of the epidemic.
We emphasize that these algorithms learn $\gbar$ which combines the effect of reporting probability and isolated population. Particularly, if a fraction $\rho$ of the total population completely isolates itself, thus getting removed from the epidemic, then $\gbar = (1-\rho)\gamma$, where $\gamma$ is the probability of reporting a case (symptomatic or asymptomatic). Hence, $\gbar$ forms the lower limit for reporting probability. In other words we can find an upper bound on total number of infected cases. Applying our algorithm on the data during the social distancing phase, we conclude with high confidence that the actual number of cases cannot be more than 35 times in New York, 40 times in Illinois, 38 times in Massachusetts, and 29 times in New Jersey, than the reported cases.
In future work, we will explore obtaining tighter bounds, when the precautions are relaxed and the fraction of isolated population $\rho$ is reduced. We will further explore how to utilize data across changing dynamics due to changing policies to strengthen these bounds.

\begin{acks}
This work was supported by National Science Foundation Award No. 2027007.
\end{acks}

\bibliographystyle{ACM-Reference-Format}
\bibliography{sample-base}


\begin{thebibliography}{13}


\ifx \showCODEN    \undefined \def \showCODEN     #1{\unskip}     \fi
\ifx \showDOI      \undefined \def \showDOI       #1{#1}\fi
\ifx \showISBNx    \undefined \def \showISBNx     #1{\unskip}     \fi
\ifx \showISBNxiii \undefined \def \showISBNxiii  #1{\unskip}     \fi
\ifx \showISSN     \undefined \def \showISSN      #1{\unskip}     \fi
\ifx \showLCCN     \undefined \def \showLCCN      #1{\unskip}     \fi
\ifx \shownote     \undefined \def \shownote      #1{#1}          \fi
\ifx \showarticletitle \undefined \def \showarticletitle #1{#1}   \fi
\ifx \showURL      \undefined \def \showURL       {\relax}        \fi
\providecommand\bibfield[2]{#2}
\providecommand\bibinfo[2]{#2}
\providecommand\natexlab[1]{#1}
\providecommand\showeprint[2][]{arXiv:#2}

\bibitem[\protect\citeauthoryear{??}{JHU}{[n.d.]}]%
        {JHUdata}
 \bibinfo{year}{[n.d.]}\natexlab{}.
\newblock \bibinfo{title}{2019 Novel Coronavirus COVID-19 (2019-nCoV) Data
  Repository by Johns Hopkins CSSE}.
\newblock
  \bibinfo{howpublished}{\url{https://github.com/CSSEGISandData/COVID-19}}.
\newblock


\bibitem[\protect\citeauthoryear{??}{dar}{[n.d.]a}]%
        {darpawinners}
 \bibinfo{year}{[n.d.]}\natexlab{a}.
\newblock \bibinfo{title}{{CHIKV} Challenge Announces Winners, Progress toward
  Forecasting the Spread of Infectious Diseases}.
\newblock
  \bibinfo{howpublished}{\url{https://www.darpa.mil/news-events/2015-05-27}}.
\newblock


\bibitem[\protect\citeauthoryear{??}{dar}{[n.d.]b}]%
        {darpachanllenge}
 \bibinfo{year}{[n.d.]}\natexlab{b}.
\newblock \bibinfo{title}{{DARPA} forecasting chikungunya challenge}.
\newblock
  \bibinfo{howpublished}{\url{https://www.innocentive.com/ar/challenge/9933617}}.
\newblock


\bibitem[\protect\citeauthoryear{??}{USp}{[n.d.]}]%
        {USpopu}
 \bibinfo{year}{[n.d.]}\natexlab{}.
\newblock \bibinfo{title}{State Population Totals: 2010-2019}.
\newblock
  \bibinfo{howpublished}{\url{https://www.census.gov/data/datasets/time-series/demo/popest/2010s-state-total.html}}.
\newblock


\bibitem[\protect\citeauthoryear{Bendavid, Mulaney, Sood, Shah, Ling,
  Bromley-Dulfano, Lai, Weissberg, Saavedra, Tedrow, et~al\mbox{.}}{Bendavid
  et~al\mbox{.}}{2020}]%
        {bendavid2020covid}
\bibfield{author}{\bibinfo{person}{Eran Bendavid}, \bibinfo{person}{Bianca
  Mulaney}, \bibinfo{person}{Neeraj Sood}, \bibinfo{person}{Soleil Shah},
  \bibinfo{person}{Emilia Ling}, \bibinfo{person}{Rebecca Bromley-Dulfano},
  \bibinfo{person}{Cara Lai}, \bibinfo{person}{Zoe Weissberg},
  \bibinfo{person}{Rodrigo Saavedra}, \bibinfo{person}{James Tedrow},
  {et~al\mbox{.}}} \bibinfo{year}{2020}\natexlab{}.
\newblock \showarticletitle{COVID-19 Antibody Seroprevalence in Santa Clara
  County, California}.
\newblock \bibinfo{journal}{\emph{MedRxiv}} (\bibinfo{year}{2020}).
\newblock


\bibitem[\protect\citeauthoryear{Bj{\o}rnstad, Finkenst{\"a}dt, and
  Grenfell}{Bj{\o}rnstad et~al\mbox{.}}{2002}]%
        {bjornstad2002dynamics}
\bibfield{author}{\bibinfo{person}{Ottar~N Bj{\o}rnstad},
  \bibinfo{person}{B{\"a}rbel~F Finkenst{\"a}dt}, {and}
  \bibinfo{person}{Bryan~T Grenfell}.} \bibinfo{year}{2002}\natexlab{}.
\newblock \showarticletitle{Dynamics of measles epidemics: estimating scaling
  of transmission rates using a time series SIR model}.
\newblock \bibinfo{journal}{\emph{Ecological monographs}} \bibinfo{volume}{72},
  \bibinfo{number}{2} (\bibinfo{year}{2002}), \bibinfo{pages}{169--184}.
\newblock


\bibitem[\protect\citeauthoryear{Coleman and Li}{Coleman and Li}{1996}]%
        {coleman1996interior}
\bibfield{author}{\bibinfo{person}{Thomas~F Coleman} {and}
  \bibinfo{person}{Yuying Li}.} \bibinfo{year}{1996}\natexlab{}.
\newblock \showarticletitle{An interior trust region approach for nonlinear
  minimization subject to bounds}.
\newblock \bibinfo{journal}{\emph{SIAM Journal on optimization}}
  \bibinfo{volume}{6}, \bibinfo{number}{2} (\bibinfo{year}{1996}),
  \bibinfo{pages}{418--445}.
\newblock


\bibitem[\protect\citeauthoryear{Ducrot, Magal, Nguyen, and Webb}{Ducrot
  et~al\mbox{.}}{2020}]%
        {ducrot2020identifying}
\bibfield{author}{\bibinfo{person}{Arnaud Ducrot}, \bibinfo{person}{P Magal},
  \bibinfo{person}{Thanh Nguyen}, {and} \bibinfo{person}{GF Webb}.}
  \bibinfo{year}{2020}\natexlab{}.
\newblock \showarticletitle{Identifying the number of unreported cases in SIR
  epidemic models}.
\newblock \bibinfo{journal}{\emph{Mathematical medicine and biology: a journal
  of the IMA}} \bibinfo{volume}{37}, \bibinfo{number}{2}
  (\bibinfo{year}{2020}), \bibinfo{pages}{243--261}.
\newblock


\bibitem[\protect\citeauthoryear{Liu, Magal, Seydi, and Webb}{Liu
  et~al\mbox{.}}{2020}]%
        {liu2020understanding}
\bibfield{author}{\bibinfo{person}{Zhihua Liu}, \bibinfo{person}{Pierre Magal},
  \bibinfo{person}{Ousmane Seydi}, {and} \bibinfo{person}{Glenn Webb}.}
  \bibinfo{year}{2020}\natexlab{}.
\newblock \showarticletitle{Understanding unreported cases in the COVID-19
  epidemic outbreak in Wuhan, China, and the importance of major public health
  interventions}.
\newblock \bibinfo{journal}{\emph{Biology}} \bibinfo{volume}{9},
  \bibinfo{number}{3} (\bibinfo{year}{2020}), \bibinfo{pages}{50}.
\newblock


\bibitem[\protect\citeauthoryear{Magal and Webb}{Magal and Webb}{2018}]%
        {magal2018parameter}
\bibfield{author}{\bibinfo{person}{Pierre Magal} {and} \bibinfo{person}{Glenn
  Webb}.} \bibinfo{year}{2018}\natexlab{}.
\newblock \showarticletitle{The parameter identification problem for SIR
  epidemic models: identifying unreported cases}.
\newblock \bibinfo{journal}{\emph{Journal of mathematical biology}}
  \bibinfo{volume}{77}, \bibinfo{number}{6-7} (\bibinfo{year}{2018}),
  \bibinfo{pages}{1629--1648}.
\newblock


\bibitem[\protect\citeauthoryear{Ricon-Becker, Tarrasch, Blinder, and
  Ben-Eliyahu}{Ricon-Becker et~al\mbox{.}}{2020}]%
        {ricon2020seven}
\bibfield{author}{\bibinfo{person}{Itay Ricon-Becker}, \bibinfo{person}{Ricardo
  Tarrasch}, \bibinfo{person}{Pablo Blinder}, {and} \bibinfo{person}{Shamgar
  Ben-Eliyahu}.} \bibinfo{year}{2020}\natexlab{}.
\newblock \showarticletitle{A seven-day cycle in COVID-19 infection and
  mortality rates: Are inter-generational social interactions on the weekends
  killing susceptible people?}
\newblock \bibinfo{journal}{\emph{medRxiv}} (\bibinfo{year}{2020}).
\newblock


\bibitem[\protect\citeauthoryear{Srivastava and Prasanna}{Srivastava and
  Prasanna}{2020}]%
        {srivastava2020learning}
\bibfield{author}{\bibinfo{person}{Ajitesh Srivastava} {and}
  \bibinfo{person}{Viktor~K Prasanna}.} \bibinfo{year}{2020}\natexlab{}.
\newblock \showarticletitle{Learning to Forecast and Forecasting to Learn from
  the COVID-19 Pandemic}.
\newblock \bibinfo{journal}{\emph{arXiv preprint arXiv:2004.11372}}
  (\bibinfo{year}{2020}).
\newblock


\bibitem[\protect\citeauthoryear{Zhou, Liu, Bai, Chen, and Wang}{Zhou
  et~al\mbox{.}}{2006}]%
        {zhou2006behaviors}
\bibfield{author}{\bibinfo{person}{Tao Zhou}, \bibinfo{person}{Jian-Guo Liu},
  \bibinfo{person}{Wen-Jie Bai}, \bibinfo{person}{Guanrong Chen}, {and}
  \bibinfo{person}{Bing-Hong Wang}.} \bibinfo{year}{2006}\natexlab{}.
\newblock \showarticletitle{Behaviors of susceptible-infected epidemics on
  scale-free networks with identical infectivity}.
\newblock \bibinfo{journal}{\emph{Physical Review E}} \bibinfo{volume}{74},
  \bibinfo{number}{5} (\bibinfo{year}{2006}), \bibinfo{pages}{056109}.
\newblock


\end{thebibliography}


\end{document}